\documentclass[conference]{IEEEtran}
\IEEEoverridecommandlockouts

\usepackage[noadjust]{cite} % Improves presentation of citations, like [3--7].
\usepackage{multirow}
\usepackage{amsmath}
\usepackage{overpic}
\usepackage{amsopn}
\usepackage{amssymb}
\usepackage{lettrine}
\usepackage{accents}
\usepackage{url}
\usepackage{mathrsfs}
\usepackage{pifont}
\usepackage{amsthm}
\usepackage{subcaption}
\usepackage{bm}         % for nice bold math symbols using \bm{...}
\usepackage{fix-cm}     % permit arbitrary sizes of Computer Modern
\usepackage{algorithm}
\usepackage{algorithmicx}
\usepackage{algpseudocode}

\makeatletter
\newcommand{\removelatexerror}{\let\@latex@error\@gobble}
\makeatother
\graphicspath{{Images/}{./}}

\usepackage{mleftright} %provides \mleft, \mright to fix issue with spacing
\mleftright                 %redefines all \left and \right to behave
\usepackage{mathdots}   % fixes \ddots and \vdots to scale properly

\makeatletter
\def\IEEElabelanchoreqn#1{\bgroup
	\def\@currentlabel{\p@equation\theequation}\relax
	\def\@currentHref{\@IEEEtheHrefequation}\label{#1}\relax
	\Hy@raisedlink{\hyper@anchorstart{\@currentHref}}\relax
	\Hy@raisedlink{\hyper@anchorend}\egroup}
\makeatother

\usepackage[utf8]{inputenc} 
\usepackage[T1]{fontenc} % To force 8-bit encoding so that
\usepackage{tikz}
\usepgflibrary{arrows.meta}
\usetikzlibrary{positioning}
\usetikzlibrary{shadows} %used also in mdframed
\usetikzlibrary{calc}
\usetikzlibrary{patterns}
\usepackage{pgfplots}
\pgfplotsset{compat=1.15}
\usepgfplotslibrary{units}
\usetikzlibrary{spy,backgrounds}
\usetikzlibrary{decorations.markings}
\usepackage{pgfplotstable}

\newtheorem{theorem}{Theorem}
\newtheorem{lemma}{Lemma}

\newtheorem{proposition}{Proposition}

\usepackage{hyphenat} % Package for hypenation-stuff

\newcommand\myprime{\mkern-0.5mu\raise0.6ex\hbox{$\scriptscriptstyle\prime$}}
\newcommand{\midk}[1]{\kern0.1em #1 \kern0.1em}
\newcommand{\middlek}[1]{\kern0.1em \middle#1 \kern0.1em}
\newcommand{\bigk}[1]{\kern-0.1em \bigm#1 \kern-0.1em}
\newcommand{\Bigk}[1]{\kern-0.1em \Bigm#1 \kern-0.1em}
\newcommand{\biggk}[1]{\kern-0.1em \biggm#1 \kern-0.1em}
\newcommand{\Biggk}[1]{\kern-0.1em \Biggm#1 \kern-0.1em}

\newcommand{\tn}[1]{\textnormal{#1}}

\newcommand{\trans}[1]{#1^{\textnormal{\textsf{\tiny T}}}}

	\newcommand{\const}[1]{\textnormal{\usefont{U}{eur}{m}{n}\selectfont #1}}
	\newcommand{\HH}{\mathop{}\!\mathsf{H}}

	\newcommand{\hh}{\mathop{}\!\const{h}}
	\newcommand{\relD}{\mathop{}\!\const{D}}
	\newcommand{\relDf}[2]{\relD\left(#1 \kern0.1em\middle\|\kern0.1em #2\right)}
	\newcommand{\erelDf}[2]{\relD(#1 \kern0.1em\|\kern0.1em #2)}
	\newcommand{\bigrelDf}[2]{\relD\bigl(#1 \kern-0.1em \bigm\| \kern-0.1em#2\bigr)}
	\newcommand{\BigrelDf}[2]{\relD\Bigl(#1 \kern-0.1em \Bigm\| \kern-0.1em#2\Bigr)}
	\newcommand{\biggrelDf}[2]{\relD\biggl(#1 \kern-0.1em \biggm\| \kern-0.1em#2\biggr)}
	\newcommand{\BiggrelDf}[2]{\relD\Biggl(#1 \kern-0.1em \Biggm\| \kern-0.1em#2\Biggr)}
	
	\newcommand{\nt}{n_{\tn{T}}}
	\newcommand{\nr}{n_{\tn{R}}}

	\newcommand{\Prvcond}[2]{\Pr\left[#1 \kern0.1em\middle|\kern0.1em #2\right]}
	\newcommand{\ePrvcond}[2]{\Pr[#1 \kern0.1em|\kern0.1em #2]}
	\newcommand{\bigPrvcond}[2]{\Pr\bigl[#1 \kern-0.1em \bigm| \kern-0.1em#2\bigr)}
	\newcommand{\BigPrvcond}[2]{\Pr\Bigl[#1 \kern-0.1em \Bigm| \kern-0.1em#2\Bigr)}
	\newcommand{\biggPrvcond}[2]{\Pr\biggl[#1 \kern-0.1em \biggm| \kern-0.1em#2\biggr)}
	\newcommand{\BiggPrvcond}[2]{\Pr\Biggl[#1 \kern-0.1em \Biggm| \kern-0.1em#2\Biggr]}
	
	\newcommand{\Prscond}[2]{\Pr\left(#1 \kern0.1em\middle|\kern0.1em #2\right)}
	\newcommand{\ePrscond}[2]{\Pr(#1 \kern0.1em|\kern0.1em #2)}
	\newcommand{\bigPrscond}[2]{\Pr\bigl(#1 \kern-0.1em \bigm| \kern-0.1em#2\bigr)}
	\newcommand{\BigPrscond}[2]{\Pr\Bigl(#1 \kern-0.1em \Bigm| \kern-0.1em#2\Bigr)}
	\newcommand{\biggPrscond}[2]{\Pr\biggl(#1 \kern-0.1em \biggm| \kern-0.1em#2\biggr)}
	\newcommand{\BiggPrscond}[2]{\Pr\Biggl(#1 \kern-0.1em \Biggm| \kern-0.1em#2\Biggr]}
	
	\newcommand{\Exp}{\operatorname{\textnormal{\textsf{E}}}}

	\newcommand{\Econd}[3][]{\Exp_{#1}\left[#2 \kern0.1em\middle|\kern0.1em #3\right]}
	\newcommand{\eEcond}[3][]{\Exp_{#1}[#2 \kern0.1em|\kern0.1em #3]}
	\newcommand{\bigEcond}[3][]{\Exp_{#1}\bigl[#2 \kern-0.1em \bigm| \kern-0.1em #3\bigr]}
	\newcommand{\BigEcond}[3][]{\Exp_{#1}\Bigl[#2 \kern-0.1em \Bigm| \kern-0.1em #3\Bigr]}
	\newcommand{\biggEcond}[3][]{\Exp_{#1}\biggl[#2 \kern-0.1em \biggm| \kern-0.1em #3\biggr]}
	\newcommand{\BiggEcond}[3][]{\Exp_{#1}\Biggl[#2 \kern-0.1em \Biggm| \kern-0.1em #3\Biggr]}
	
	\newcommand{\dd}{\mathop{}\!\mathrm{d}}

	\newcommand{\supp}{\operatorname{supp}}

	\makeatletter
	
	\newcommand{\Rmnum}[1]{\expandafter\@slowromancap\romannumeral #1@}
	\makeatother
	
	\usepackage{threeparttable}
	\usepackage{diagbox}
	\usepackage{array}
	\usepackage{boldline}
	\usepackage{caption}
	\usepackage{dsfont}
	
	\usepackage{algpseudocode}
	
	\allowdisplaybreaks
	\usetikzlibrary{datavisualization}
	
	\usetikzlibrary{spy,backgrounds}
	\usetikzlibrary{decorations.markings}

	\def\BibTeX{{\rm B\kern-.05em{\sc i\kern-.025em b}\kern-.08em
			T\kern-.1667em\lower.7ex\hbox{E}\kern-.125emX}}
	
\begin{document}
		
		\title{Performance Evaluation of A Certain Transceiver Architecture for Multiple-Input Multiple-Output Phase-Modulated Channels\\
			\thanks{This work is supported in part by the National Natural Science Foundation of China under Grant No. 62401593 and 12374275, the Young Elite Scientists Sponsorship Program by the China Association for Science and Technology, and Innovation Research Foundation of National University of Defense Technology. (Corresponding Authors: Ru-Han Chen and Shijun Zhu)
			}
		}
		
		\author{\IEEEauthorblockN{1\textsuperscript{st} Hengyu Cui}
			\IEEEauthorblockA{\textit{Department of Information Physics and Engineering} \\
				\textit{Nanjing University of Science and Technology}\\
				Nanjing, China\\
				1719170887@qq.com}
			\and
			
			\IEEEauthorblockN{2\textsuperscript{nd} Ru-Han Chen}
			\IEEEauthorblockA{\textit{Sixty-Third Research Institute} \\
				\textit{National University of Defense Technology}\\
				Nanjing, China \\
				tx\_rhc22@nudt.edu.cn
			}
			\and
			
			\IEEEauthorblockN{3\textsuperscript{rd} Zhenyao He}
			\IEEEauthorblockA{\textit{Sixty-Third Research Institute} \\
				\textit{National University of Defense Technology}\\
				Nanjing, China \\
				zyhe1111@nudt.edu.cn
			}
			\and
			
			\IEEEauthorblockN{4\textsuperscript{th} Shijun Zhu}
			\IEEEauthorblockA{\textit{Department of Information Physics and Engineering} \\
				\textit{Nanjing University of Science and Technology}\\
				Nanjing, China \\
				shijunzhu@njust.edu.cn
			}
			\and
			
			\IEEEauthorblockN{5\textsuperscript{th} Ruoqi Sun}
			\IEEEauthorblockA{\textit{Sixty-Third Research Institute} \\
				\textit{National University of Defense Technology}\\
				Nanjing, China \\
				202412491635@nuist.edu.cn
			}
			\and
			
			\IEEEauthorblockN{6\textsuperscript{th} Yeqin Tai}
			\IEEEauthorblockA{\textit{Department of Information Physics and Engineering} \\
				\textit{Nanjing University of Science and Technology}\\
				Nanjing, China \\
				taiyeqin1412@163.com
			}
		}
		
		\maketitle
		
		\begin{abstract}
			For multiple-input multiple-output (MIMO) channels with phase modulation, we recently proposed a method of unitarily transforming the channel matrix into a certain row-echelon form, by which the original MIMO channel can be converted into a certain number of scalar sub-channels with two phase inputs, thereby forming an annulus constellation geometry, and corrupted by both the additive white Gaussian noise and weak self-interference. In this paper, several bounds are derived to evaluate the fundamental limit of such a specific transceiver architecture. Two upper bounds are obtained by upper-bounding the capacity of a scalar channel with an annulus support constraint from the perspective of the convex geometry, while a lower bound is obtained by the standard entropy power inequality. Numerical results show that the gaps between these bounds are small at high signal-to-noise ratios for the MIMO phase-modulated channels over the Rayleigh fading and the single-input multiple-output symbiotic communication system assisted by a reconfigurable intelligent surface.
		\end{abstract}
		
		
		\section{Introduction}
		
		In multiple-input multiple-output (MIMO) communication systems, the constant-envelope transmitted signal can significantly improve the efficiency of power amplifiers, which highlights the importance of phase modulation \cite{computing2005capacity,yousefbeiki2014efficient}. Meanwhile, to reduce the implementation cost, the concept of the single radio-frequency (RF) MIMO has been proposed in \cite{li2021single,karasik2021single,chen2024when}, by which a MIMO phase-modulated communication system can be equivalently realized by using only one RF chain at the transmitter and temporal-spatial coding of a low-cost reconfigurable intelligent surface (RIS).
		
		Different from the classic MIMO channel subject to a total average-power constraint, the MIMO phase-modulated channel is governed by the non-convex support constraint, since the information can only be modulated onto the phases of input signals. To the best of the authors' knowledge, the degrees of freedom (DoFs) of this channel are exactly characterized in \cite[Thm. 2]{cheng2021dof}, whereas its capacity remains unknown. Existing relevant results are limited to special cases, e.g. scalar phase-modulated channels \cite{wyner66_1} or certain discrete alphabets \cite{He2005TIT}. In our previous work \cite{chen2024when}, we have derived the achievable rate of the transceiver based on the QR decomposition and the successive interference cancellation (SIC), by which the MIMO channel is decomposed into several parallel sub-channels. Since there is only one phase input in each sub-channel (except for the last one), almost half DoFs will be lost in the regime with strong spatial correlation. Fortunately, in our recent work \cite{sun2026grouped}, we show that such a drawback can be overcome by a low-complexity matrix decomposition method, termed the combinatorially pairing (CP) algorithm, which transforms the channel matrix into a certain row-echelon form such that two independent phase inputs are grouped in each sub-channel (also except for the last one), thus attaining almost full DoFs. Lack of performance evaluation on this novel transceiver for the MIMO phase-modulated channel motivates this work.
		
		In this paper, we analyze the achievable rate of the above-mentioned transceiver architecture. Main technical challenges arise from two aspects: i) the residual error caused by the CP decomposition yields non-Gaussian self-inteference; and ii) there are two phase-modulated inputs in each sub-channel, which corresponds to a non-convex constellation geometry. Consequently, the existing capacity analysis methods \cite{smith1971,mckellips2004,thangaraj2017,rassouli2016}, developed for the channels with a convex input support or the additive white Gaussian noise (AWGN), can not be applied to our considered channel.
		
		The main contributions of this work are as follows:
		\begin{itemize}
			\item By characterizing the non-Gaussian noise and the annular support constraints, two upper bounds are derived from the perspective of the convex geometry, while a lower bound is obtained by the entropy power inequality (EPI).
			\item We further apply the above theoretical framework to two scenarios. Numerical results show that the derived upper and lower bounds are remarkably close to each other.
		\end{itemize}
		
		\section{System Model and Problem Formulation}
		
		\subsection{System Model}
		
		We consider a MIMO phase-modulated channel with $\nt$ transmit antennas and $\nr$ receive antennas as follows
		\begin{equation}
			\mathbf{Y} = \sqrt{\mathrm{snr}}\,\mathsf{H}\cdot \exp(j\bm{\Theta}) + \mathbf{Z},
		\end{equation}
		where the $\nr$-dimensional complex-valued vector $\mathbf{Y}$ denotes the channel input, the $\nt$-dimensional complex-valued vector $\exp(j\bm{\Theta}) \in \mathbb{C}^{\nt}$ denotes $\nt$ phase-modulated inputs, the $\nr\times \nt$ complex-valued matrix $\mathsf{H} \in \mathbb{C}^{\nr \times \nt}$ consists of channel gains from the transmitter to the receiver, $\nr$-dimensional complex-valued random vector $\mathbf{Z}\sim \mathcal{CN}(\bm{0}_{\nr},\mathsf{I}_{\nr})$ denotes the additive noise, and the positive scaling factor $\mathrm{snr}$ is used to measure the signal-to-noise ratio (SNR).
		
		\subsection{Combinatorially Pairing Algorithm}
		
		For the MIMO channel, a classic type of transceiver architectures is unitarily transforming the channel matrix into some row-echolon form, and then decoding these row sub-channels from the bottom to top by the SIC. Unlike the conventionally-used QR decomposition, the CP algorithm recently proposed in \cite{sun2026grouped}, unitarily transforms the MIMO channel matrix into a special row-echelon form with very small residual error, which groups two inputs in each sub-channel (except for the last one). Due to space limitations, we refer the reader interested in the CP algorithm to \cite[Sec.~IV]{sun2026grouped}. Instead, we directly express $\nr$ scalar sub-channels as follows
		\begin{equation}\label{sys_model_sub}
			\tilde{Y}_{i} = \sqrt{\mathrm{snr}}\,\left(c_{i,2i-1}e^{j\Theta_{2i-1}} + c_{i,2i}e^{j\Theta_{2i}}\right) + \tilde{N}_{i},
		\end{equation}
		for $i\in\{1,2,\dots,\nr\}$, where the total noise
		\begin{flalign}
			\tilde{N}_{i} = \sqrt{\mathrm{snr}}\Bigg(\sum_{k=1}^{2i-2}c_{i,k}e^{j\Theta_k}\Bigg)+\tilde{Z}_i,
		\end{flalign}
		takes the self-interference caused by the residual error of the CP algorithm into account, $\tilde{Z}_i\sim\mathcal{CN}(0,1)$ denotes the noise term transformed by $\mathbf{Z}$, and $2i$ complex-valued coefficients $\{c_{i,k}\}_{k=1}^{2i}$ are outputted from the CP algorithm. Due to the existence of the self-interference term $\sum_{k=1}^{2i-2}c_{i,k}e^{j\Theta_k}$, the total noise $\tilde{N}_i$ is non-Gaussian distributed. Nevertheless, its variance, denoted by $\mathcal{E}_i$, can be explicitly upper-bounded as follows
		\begin{flalign}
			\mathcal{E}_{i}
			&=\mathsf{var}\Bigg[\sqrt{\mathrm{snr}}\Bigg(\sum_{k=1}^{2i-2}c_{i,k}e^{j\Theta_k}\Bigg)\Bigg]+1\\
			&\le \mathrm{snr}\cdot\left(\sum_{k=1}^{2i-2}|c_{i,k}|^2\right)+1,\label{eq:var}
		\end{flalign}
		where Eq.~\eqref{eq:var} follows from the statistical independence among all phase-modulated inputs $\exp(j\Theta_1),\ldots,\exp(j\Theta_{\nt})$.
		
		For notational simplicity, we define the equivalent input for the $i$-th sub-channel as follows
		\begin{equation}\label{S_i}
			S_i \triangleq \sqrt{\mathrm{snr}}\,\left(c_{i,2i-1}e^{j\Theta_{2i-1}} + c_{i,2i}e^{j\Theta_{2i}}\right).
		\end{equation}
		According to \cite[Lemma 1]{sun2026grouped}, the sum of two phase-modulated components lies within a two-dimensional annulus. Consequently, the equivalent input signal $S_i$ is supported on the annulus $\mathcal{P}_i=\{S_i\in\mathbb{C}\mid r_i\le|S_i|\le R_i\}$, where the inner and the outer radii are given by
		\begin{subequations}
			\begin{align}
				r_i &= \sqrt{\mathrm{snr}}\,\bigl||c_{i,2i-1}|-|c_{i,2i}|\bigr|,\label{eq:ri}\\
				R_i &= \sqrt{\mathrm{snr}}\,\bigl(|c_{i,2i-1}|+|c_{i,2i}|\bigr).\label{eq:Ri}
			\end{align}
		\end{subequations}
		We refer the reader interested in the difference between signaling over a single phase input and that over two phase inputs, to \cite[Fig. 2]{sun2026grouped}, which intuitively explains why the transceiver based on the CP algorithm can obtain higher DoFs.
		
		Since, in the considered transceiver architecture, the equivalent input $S_i$ is only decoded from the $i$-th sub-channel \eqref{sys_model_sub}, the total achievable rate of the considered transceiver architecture is given as
		\begin{flalign}
			\const{R}_{\text{total}}=\sum_{i=1}^{\nr}\max_{r_i\le|S_i|\le R_i}I(S_i;\tilde{Y}_i).
		\end{flalign}
		
		\subsection{Problem Formulation}
		
		For a general description of our result, we drop the subscript $i$ and consider the scalar complex-valued channel with non-Gaussian additive noise as follows
		\begin{equation}\label{sub-channel}
			Y=S+N,
		\end{equation}
		where the complex-valued input $S$ satisfies the annulus support constraint $|S|\in[r,R]$ and the additive noise $N$ has zero mean and variance $\mathcal{E}$ and is independent of the input $S$.
		
		Clearly, bounds on the total achievable rate $\const{R}_{\text{total}}$ of the considered transceiver can be obtained by upper-bounding and lower-bounding the channel capacity of the scalar channel in \eqref{sub-channel}, denoted as $\const{C}(r,R,\mathcal{E})$, which is defined as
		\begin{flalign}
			\const{C}(r,R,\mathcal{E})
			&\triangleq\max_{r\le|S|\le R}I(S;Y)\\
			&=\max_{r\le|S|\le R}\hh(Y)-\hh(N).\label{eq:capacity}
		\end{flalign}
		
		\section{Bounds on $\const{C}(r,R,\mathcal{E})$}
		
		To avoid the ambiguity, we regard the complex-valued channel output $Y$ as the two-dimensional real-valued random vector $\bar{\mathbf{Y}}$ consisting of the real and the imaginary parts of $Y$. The corresponding two-dimensional input $\bar{\mathbf{S}}$ is accordingly subject to an annulus support constraint as follows
		\begin{flalign}
			\supp\bar{\mathbf{S}}\subseteq\mathcal{A}
			\triangleq
			\left\{
			\mathbf{x}\in\mathbb{R}^2:r\le\left\|\mathbf{x}\right\|_2\le R
			\right\}.
		\end{flalign}
		The corresponding two-dimensional real-valued noise term is denoted by $\bar{\mathbf{N}}=\trans{(\bar{N}_1,\bar{N}_2)}$, which has zero mean and satisfies $\mathbb{E}[\bar{N}_1^2+\bar{N}_2^2]\le\mathcal{E}$.
		
		Since the differential entropy $\hh(\bar{\mathbf{N}})$ of the noise term $\bar{\mathbf{N}}$ in the channel \eqref{sub-channel} is independent of the channel input $\bar{\mathbf{S}}$, the bounds on the channel capacity $\const{C}(r,R,\mathcal{E})$ can be obtained by upper-bounding or lower-bounding the differential entropy $\hh(\bar{\mathbf{Y}})$.
		
		\subsection{EPI-Based Lower Bound}
		
		We derive a lower bound on the channel capacity $\const{C}(r,R,\mathcal{E})$ via the EPI.
		
		\begin{theorem}\label{thm:1}
			For the channel~\eqref{sub-channel}, its capacity is lower-bounded as
			\begin{equation}\label{eq:EPI_lower_bound}
				\const{C}(r,R,\mathcal{E})\geq
				\log\left(
				1+\frac{\pi(R^2-r^2)}{\exp\big(\hh(N)\big)}
				\right),
			\end{equation}
			where $\hh(N)$ represents the differential entropy of $N$.
		\end{theorem}
		
		\begin{IEEEproof}
			Let $\bar{\mathbf{S}}$ be uniformly distributed on the annular region $\mathcal{A}$. Its differential entropy is given as
			\begin{equation}\label{eqn:input_entropy}
				\hh(\bar{\mathbf{S}})=\log\bigl(\pi(R^2-r^2)\bigr).
			\end{equation}
			
			The EPI for two independent two-dimensional random vectors states
			\begin{equation}
				e^{\hh(\bar{\mathbf{Y}})}
				\ge e^{\hh(\bar{\mathbf{S}})}+e^{\hh(\bar{\mathbf{N}})}.
			\end{equation}
			Hence, we have
			\begin{flalign}
				I(S;Y)
				&=\hh(Y)-\hh(N)\\
				&\ge\log\bigl(e^{\hh(\bar{\mathbf{S}})}+e^{\hh(\bar{\mathbf{N}})}\bigr)-\hh(\bar{\mathbf{N}})\\
				&=\log\left(1+e^{\hh(\bar{\mathbf{S}})-\hh(\bar{\mathbf{N}})}\right).\label{eqn:thm1-1}
			\end{flalign}
			Substituting Eq.~\eqref{eqn:input_entropy} into Eq.~\eqref{eqn:thm1-1}, we conclude the theorem.
		\end{IEEEproof}
		
		\subsection{Capacity Upper Bounds}
		
		Since the channel output $\bar{\mathbf{Y}}$ can be expressed as the sum of two independent random variables (i.e., the channel input $\bar{\mathbf{S}}$ and the noise term $\bar{\mathbf{N}}$), its differential entropy can be upper-bounded from the perspective of the convex geometry \cite{jog2016geometric}. The below approach is also applicable to other support constraints on the channel input.
		
		We first note that the differential entropy can be related with the $2n$-dimensional volume of the typical set $\mathcal{T}^{\epsilon}_n$ of $\bar{\mathbf{Y}}$ (with respect to $n$ channels uses) as follows \cite[Eq. (4.5)]{jog2016geometric}
		\begin{flalign}
			\hh(\bar{\mathbf{Y}})
			=
			\lim_{\epsilon\to0_{+}}
			\lim_{n\to+\infty}
			\frac{1}{n}
			\log\left(
			\mathsf{vol}_{2n}\left(
			\mathcal{T}^{\epsilon}_n
			\right)
			\right).
		\end{flalign}
		Note that the channel input is subject to a support constraint $\supp\bar{\mathbf{S}}\subseteq\mathcal{A}$ and the noise term satisfies $\mathbb{E}[\bar{N}_1^2+\bar{N}_2^2]\le\mathcal{E}$. Then the volume of the typical set $\mathcal{T}^{\epsilon}_n$ can be bounded as follows
		\begin{flalign}\label{eq:bound_typical_set}
			\mathsf{vol}_{2n}\left(
			\mathcal{T}^{\epsilon}_n
			\right)
			\le
			\mathsf{vol}_{2n}\left(
			\mathcal{A}^n\oplus\mathcal{B}_{2n}(\sqrt{n\mathcal{E}})
			\right),
		\end{flalign}
		where the notation $\oplus$ denotes the Minkowski summation, and $\mathcal{B}_{2n}(\sqrt{n\mathcal{E}})$ denotes the $2n$-ball of radius $\sqrt{n\mathcal{E}}$.
		
		To evaluate the volume of the typical set, we reformulate an existing result as follows.
		
		\begin{lemma}[{\cite[Eq.~(3.167)]{Jog2014thesis}}]\label{lemma:1}
			For any convex body $\mathcal{K}\subseteq\mathbb{R}^d$ with the intrinsic volumes $\alpha_0,\ldots,\alpha_{d-1}$, and $\alpha_d$ (i.e., equal to the $d$-dimensional volume of $\mathcal{K}$), define the exponential growth rate function
			\begin{flalign}\label{eq:growth_rate1}
				\ell_{\mathcal{K}}(\nu)
				\triangleq
				\limsup_{n\to\infty}
				\frac{1}{n}
				\log\left(
				\mathsf{vol}_{nd}
				\left(
				\mathcal{K}^n\oplus
				\mathcal{B}_{nd}\bigl(\sqrt{nd\nu}\bigr)
				\right)
				\right),
			\end{flalign}
			for all $\nu>0$, and the logarithmic moment generating function
			\begin{flalign}\label{eq:LMG}
				\Lambda_{\mathcal{K}}(t)
				\triangleq
				\log\left(
				\sum_{k=0}^{d}\alpha_k\exp(kt)
				\right).
			\end{flalign}
			Then
			\begin{flalign}\label{eq:growth_rate2}
				\ell_{\mathcal{K}}(\nu)
				=
				\sup_{\theta\in[0,1]}
				\left(
				-\Lambda_{\mathcal{K}}^*(d\theta)
				+
				\frac{d(1-\theta)}{2}
				\log\frac{2\pi e\nu}{1-\theta}
				\right),
			\end{flalign}
			where $\Lambda_{\mathcal{K}}^*(\cdot)$ denotes the convex conjugate of the logarithmic moment generating function $\Lambda_{\mathcal{K}}(\cdot)$.
		\end{lemma}
		
		Note that the original support set $\mathcal{A}$ is non-convex, and hence, Lemma~\ref{lemma:1} can not be directly applied to upper-bound the right-hand side in Eq.~\eqref{eq:bound_typical_set}. To tackle this technical challenge, two different methods are given as follows.
		
		\subsubsection{Upper Bound via Support Relaxation}
		
		By relaxing the annular support constraint to $\supp\bar{\mathbf{S}}\subseteq\mathcal{D}$, where $\mathcal{D}$ denotes the two-dimensional closed disk as follows:
		\begin{flalign}
			\mathcal{D}
			\triangleq
			\left\{
			\mathbf{x}\in\mathbb{R}^2:
			\left\|\mathbf{x}\right\|_2\le R
			\right\}.
		\end{flalign}
		Due to the convexity of $\mathcal{D}$, we can directly apply the Lemma~\ref{lemma:1} to obtain the following result.
		
		\begin{lemma}
			Define
			\begin{flalign}
				t^{\star}(\theta)
				\triangleq
				\log\left(
				-\frac{1}{2R}
				+\frac{
					1+\sqrt{4(1-\frac{4}{\pi})\theta^2-4(1-\frac{4}{\pi})\theta+1}
				}{
					4R(1-\theta)
				}
				\right)
				\label{eq:t_theta}
			\end{flalign}
			and
			\begin{flalign}
				g_{\nu}(\theta)
				=&-2\theta t^{\star}(\theta)
				+\log\left(
				1+\pi Re^{t^{\star}(\theta)}
				+\pi R^2e^{2t^{\star}(\theta)}
				\right)
				\nonumber\\
				&+(1-\theta)\log\frac{2\pi e\nu}{1-\theta}
			\end{flalign}
			for $\theta\in(0,1)$, and let $g_{\nu}(0)=\log(2\pi e\nu)$ and $g_{\nu}(1)=\log(\pi R^2)$. The exponential growth rate function of the closed disk $\mathcal{D}$ is given by
			\begin{equation}
				\ell_{\mathcal{D}}(\nu)
				=
				\max_{\theta\in[0,1]}g_{\nu}(\theta),
				\quad\nu>0.
			\end{equation}
		\end{lemma}
		
		\begin{IEEEproof}
			From \cite[Example 6.1.2]{Lotz2020}, we know that intrinsic volumes of the two-dimensional convex disk $\mathcal{A}$ of radius $R$ are given as the Euler characteristic $\alpha_0=1$, the semi-perimeter $\alpha_1=\pi R$, and the area $\alpha_2=\pi R^2$. By the definition \eqref{eq:LMG}, we have
			\begin{equation}
				\Lambda_{\mathcal{D}}(t)
				=
				\log\left(
				1+\pi Re^t+\pi R^2e^{2t}
				\right),
			\end{equation}
			for all $t\in\mathbb{R}$, whose conjugate function is given as
			\begin{flalign}
				\Lambda_{\mathcal{D}}^{\ast}(\lambda)
				=
				\sup_{t\in\mathbb{R}}
				\lambda t-\Lambda_{\!\mathcal{D}}(t)
			\end{flalign}
			for all $\lambda\in[0,2]$. Instead of characterizing the closed-form expression of the conjugate function $\Lambda_{\mathcal{D}}^{\ast}(\cdot)$, we directly consider the optimization problem involved in Eq.~\eqref{eq:growth_rate2} in what follows
			\begin{flalign}
				&~\ell_{\mathcal{D}}(\nu)\nonumber\\
				=&\sup_{\theta\in[0,1]}
				\left(
				-\Lambda_{\mathcal{D}}^*(2\theta)
				+(1-\theta)
				\log\frac{2\pi e\nu}{1-\theta}
				\right),\\
				=&\sup_{\theta\in[0,1]}
				\left(
				-\left(
				\sup_{t\in\mathbb{R}}
				2\theta t-\Lambda_{\!\mathcal{D}}(t)
				\right)
				+(1-\theta)
				\log\frac{2\pi e\nu}{1-\theta}
				\right).
				\label{eq:lemma2-1}
			\end{flalign}
			
			It can be easily verified that the objective function
			$\zeta_{\theta}(t)=-(2\theta t-\Lambda_{\!\mathcal{D}}(t))$
			is strongly convex in $t\in\mathbb{R}$. Note that its derivative is given as
			\begin{flalign}
				\frac{\dd\zeta_{\theta}}{\dd t}
				=
				-2\theta
				+
				\frac{
					\pi Re^t+2\pi R^2e^{2t}
				}{
					1+\pi Re^t+\pi R^2e^{2t}
				}.
			\end{flalign}
			Clearly, $\frac{\dd\zeta_0}{\dd t}>0$ and $\frac{\dd\zeta_1}{\dd t}<0$ for all $t\in\mathbb{R}$. Thus, we know that
			\begin{subequations}
				\begin{align}
					\inf_{t\in\mathbb{R}}\zeta_0(t)
					&=
					\lim_{t\to-\infty}\zeta_0(t)
					=g_{\nu}(0),
					\label{eq:prop_2_3}\\
					\inf_{t\in\mathbb{R}}\zeta_1(t)
					&=
					\lim_{t\to+\infty}\zeta_1(t)
					=g_{\nu}(1).
					\label{eq:prop_2_4}
				\end{align}
			\end{subequations}
			For $\theta\in(0,1)$, by the condition on the extreme point, i.e., ${\dd\zeta_{\theta}}/{\dd t}=0$, we know that $\zeta_{\theta}(t)$ is minimized at $t=t^{\star}(\theta)$ (see Eq.~\eqref{eq:t_theta}), and hence $\inf_{t\in\mathbb{R}}\zeta_{\theta}(t)=g_{\nu}(\theta)$. It is straightforward to prove the continuity of $g_{\nu}(\theta)$ on $[0,1]$, thereby concluding the lemma.
		\end{IEEEproof}
		
		\begin{proposition}\label{cor:disc_bound}
			For the channel~\eqref{sub-channel}, the capacity is upper bounded as
			\begin{equation}\label{eq:disc_ub}
				\const{C}(r,R,\mathcal{E})
				\le
				\ell_{\mathcal{D}}\left(\frac{\mathcal{E}}{2}\right)-\hh(N).
			\end{equation}
		\end{proposition}
		
		\begin{IEEEproof}
			It follows from Eq.~\eqref{eq:growth_rate1} and
			$\mathsf{vol}_{2n}(\mathcal{T}^{\epsilon}_n)
			\le
			\mathsf{vol}_{2n}
			(\mathcal{D}^n\oplus\mathcal{B}_{2n}(\sqrt{n\mathcal{E}}))$
			that
			\begin{flalign}
				\hh(\bar{\mathbf{Y}})
				&=
				\lim_{\epsilon\to0_{+}}
				\lim_{n\to+\infty}
				\frac{1}{n}
				\log\left(
				\mathsf{vol}_{2n}
				\left(
				\mathcal{T}^{\epsilon}_n
				\right)
				\right)\\
				&\le
				\ell_{\mathcal{D}}\left(\mathcal{E}/2\right).
				\label{eq:proof_1}
			\end{flalign}
			The proof is completed by combining Eqs.~\eqref{eq:proof_1} and \eqref{eq:capacity}.
		\end{IEEEproof}
		
		\subsubsection{Upper Bound via Polar Decomposition}
		
		We switch our gear back to the complex-valued channel output $Y$, which can be written as $Y=R_Y\cdot\exp(j\Theta_Y)$ via the polar decomposition, where $R_Y$ and $\Theta_Y$ is the modulus and the argument of $Y$. It can be shown that the differential entropy $\hh(Y)$ satisfies
		\begin{flalign}
			\hh(Y)
			=
			\hh(R_Y)
			+
			\hh(\Theta_Y|R_Y)
			+
			\mathbb{E}[\log R_Y],
		\end{flalign}
		which can be easily proved by the change of variables. It is straightforward to see
		\begin{flalign}
			\hh(\Theta_Y|R_Y)
			\le
			\hh(\Theta_Y)
			\le
			\log(2\pi),
		\end{flalign}
		which immediately leads to
		\begin{flalign}\label{eq:ubd_polar}
			\hh(Y)
			\le
			\hh(R_Y)
			+
			\mathbb{E}[\log R_Y]
			+
			\log(2\pi).
		\end{flalign}
		
		One can regard the modulus variable $R_Y=|S+N|$ as the sum of the modulus of the input $S$, denoted by the nonnegative random variable $R_S$, and some signal-dependent noise term $R_N=R_Y-R_S$ with bounded variance. In the following proposition, we present an upper bound on $\hh(R_Y)$ via some useful statistical properties of real-valued random variables $R_S$ and $R_N$.
		
		\begin{proposition}\label{prop:2}
			Let $\mathcal{L}\triangleq[r,R]$, whose exponential growth rate function is denoted by
			\begin{flalign}\label{eq:growth_rate_L}
				\ell_{\mathcal{L}}(\nu)
				&\triangleq
				\limsup_{n\to\infty}
				\frac{1}{n}
				\log\left(
				\mathsf{vol}_{n}
				\left(
				\mathcal{L}^n
				\oplus
				\mathcal{B}_{n}\bigl(\sqrt{n\nu}\bigr)
				\right)
				\right).
			\end{flalign}
			The random variables $R_S$ and $R_N$ satisfy
			\begin{subequations}
				\begin{align}
					\supp R_S
					&\subseteq\mathcal{L},
					\label{eq:prop_2_1}\\
					\mathsf{var}(R_N)
					&\le\mathcal{E}.
					\label{eq:prop_2_2}
				\end{align}
			\end{subequations}
			and the differential entropy of $R_Y$ can be upper-bounded as follows
			\begin{flalign}\label{eq:diff_ent_RY}
				\hh(R_Y)\le\ell_{\mathcal{L}}(\mathcal{E}).
			\end{flalign}
		\end{proposition}
		
		\begin{IEEEproof}
			Eq.~\eqref{eq:prop_2_1} follows from the annulus support constraint $|S|\in[r,R]$. By the triangle inequality, we have
			\begin{flalign}
				|R_N|
				&=
				|R_Y-R_S|\\
				&\le|N|,
			\end{flalign}
			thus leading to
			$\mathsf{var}(R_N)\le\mathbb{E}[|R_N|^2]\le\mathbb{E}[|N|^2]=\mathcal{E}$.
			Then the upper bound \eqref{eq:diff_ent_RY} can be derived by using \cite[Thm. 3.7.7]{Jog2014thesis}.
		\end{IEEEproof}
		
		The following lemma characterizes the exponential growth rate function for closed intervals.
		
		\begin{lemma}[{\cite[Theorem~3.7.1]{Jog2014thesis}}]\label{lem:jog_core}
			Let $\mathcal{A}_{L}$ be a closed interval of length $L$. Its exponential growth rate function is given by
			\begin{flalign}\label{eq:growth_rate_A}
				\ell_{\mathcal{A}_{L}}(\nu)
				&\triangleq
				\limsup_{n\to\infty}
				\frac{1}{n}
				\log\left(
				\mathsf{vol}_{n}
				\left(
				\mathcal{A}_{L}^n
				\oplus
				\mathcal{B}_{n}\bigl(\sqrt{n\nu}\bigr)
				\right)
				\right),
				\\
				&=
				\HH_b(\theta^{\ast})
				+
				(1-\theta^{\ast})\log(L)
				+
				\frac{\theta^{\ast}}{2}
				\log\frac{2\pi e\nu}{\theta^{\ast}}
			\end{flalign}
			where the binary entropy function
			$\HH_b(\theta)=-\theta\log\theta-(1-\theta)\log(1-\theta)$,
			and $\theta^\ast\in[0,1]$ is the unique solution to the following equation
			\begin{equation}
				\frac{(1-\theta^\ast)^2}{\theta^{\ast3}}
				=
				\frac{L^2}{2\pi\nu}.
			\end{equation}
		\end{lemma}
		
		Then, another upper bound on $\const{C}(r,R,\mathcal{E})$ is given as follows.
		
		\begin{proposition}\label{cor:amp_bound}
			For the channel~\eqref{sub-channel}, the capacity is upper bounded as
			\begin{flalign}\label{eq:amp_ub}
				&\const{C}(r,R,\mathcal{E})
				\nonumber\\
				&~
				\le
				\ell_{\mathcal{A}_{R-r}}(\mathcal{E})
				+
				\log(2\pi)
				+
				\frac{1}{2}\log(R^2+\mathcal{E})
				-
				\hh(N).
			\end{flalign}
		\end{proposition}
		
		\begin{IEEEproof}
			By Proposition~\ref{prop:2}, we obtain
			$\hh(R_Y)\le\ell_{\mathcal{A}_{R-r}}(\mathcal{E})$
			since the length of $\mathcal{L}$ is $R-r$.
			Next, using Jensen's inequality and the Cauchy-Schwarz inequality to bound the expectation term $\mathbb{E}[\log R_Y]$ as follows:
			\begin{equation}
				\mathbb{E}[\log R_Y]
				=
				\frac{1}{2}\mathbb{E}[\log R_Y^2]
				\leq
				\frac{1}{2}\log\mathbb{E}[R_Y^2].
			\end{equation}
			Since
			$\mathbb{E}[R_Y^2]=\mathbb{E}[|S+N|^2]\leq R^2+\mathcal{E}$,
			we obtain
			$\mathbb{E}[\log R_Y]\leq\frac{1}{2}\log(R^2+\mathcal{E})$.
			We conclude the proposition by combining our derived upper bounds on each term for the right-hand side term in Eq.~\eqref{eq:ubd_polar}.
		\end{IEEEproof}
		
		\section{Numerical Results}
		
		This section is devoted to numerical verification of our derived bounds.
		
		\subsection{Scalar AWGN Channels}
		
		We first consider the scalar AWGN channel with an annular input constraint, where the noise variance is exactly $\mathcal{E}$. In Fig.~\eqref{fig:bounds}, we compare the upper bounds~\eqref{eq:disc_ub} and~\eqref{eq:amp_ub} with the EPI lower bound~\eqref{eq:EPI_lower_bound} for a narrow annulus ($r=0.8$ and $R=1$) and a wide annulus ($r=0.2$ and $R=1$). It can be seen that, in each case, our derived upper bounds gradually approach the EPI-based lower bound as the SNR increases. It can also be noted that, for the wide annulus, the upper bound~\eqref{eq:disc_ub} is better than the other, since the loss of support relaxation is relatively small.
		
		\begin{figure}[htbp]
			\centering
			\begin{subfigure}[b]{0.67\columnwidth}
				\centering
				\includegraphics[width=\linewidth]{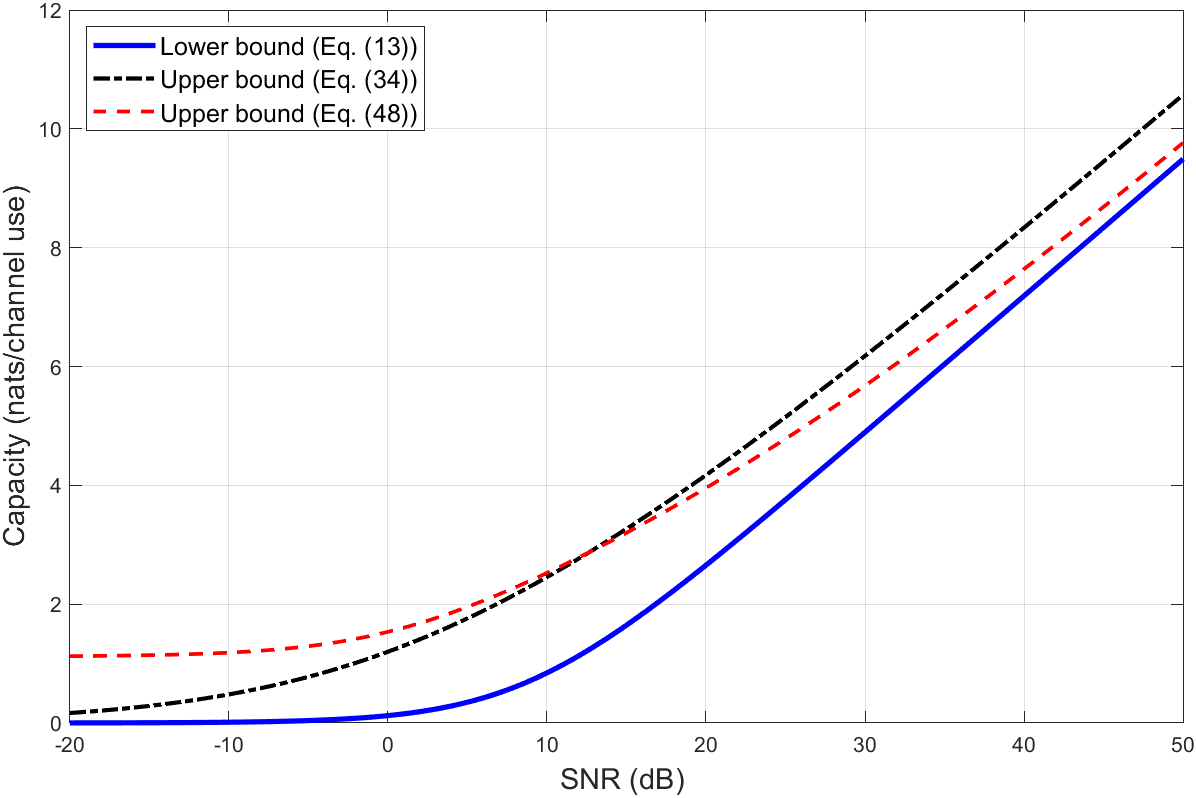}
				\caption{$r=0.8$ and $R=1$ (narrow annulus)}
				\label{fig:left}
			\end{subfigure}
			\hfill
			\begin{subfigure}[b]{0.67\columnwidth}
				\centering
				\includegraphics[width=\linewidth]{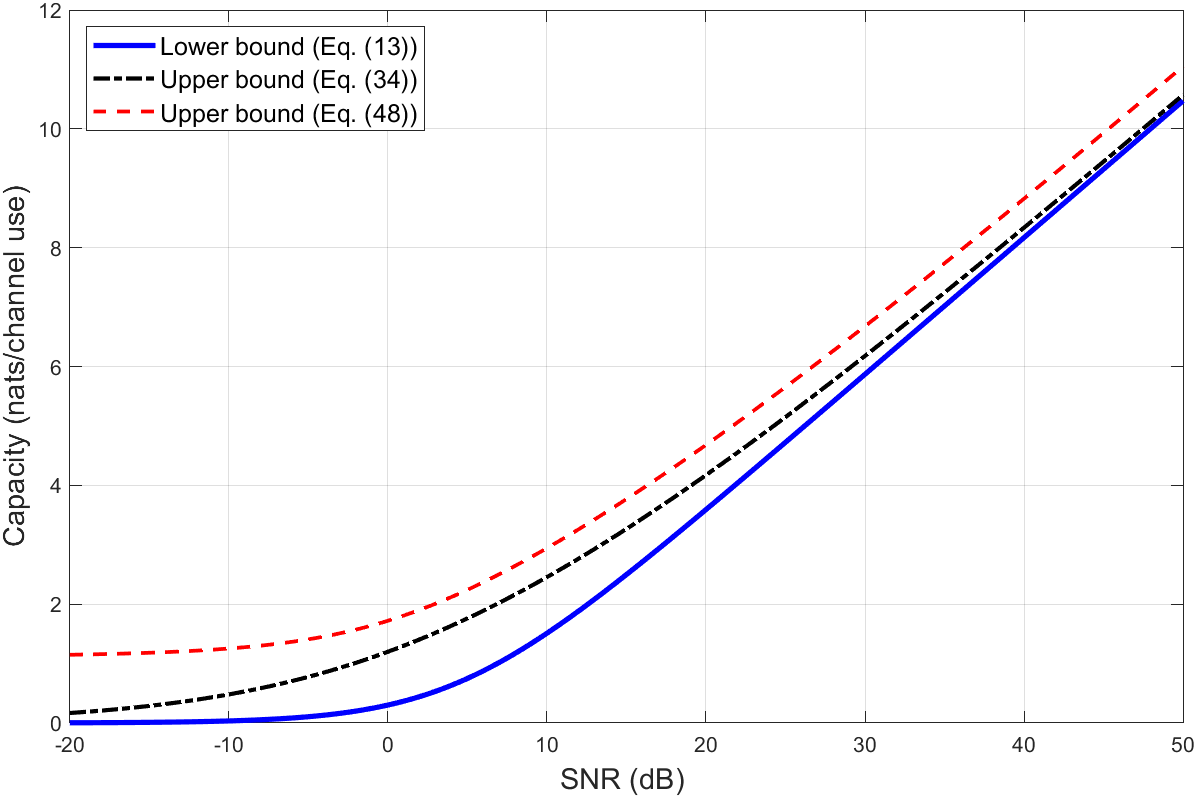}
				\caption{$r=0.2$ and $R=1$ (wide annulus)}
				\label{fig:right}
			\end{subfigure}
			\caption{Capacity bounds for scalar AWGN channels with annulus constellation geometry.}
			\label{fig:bounds}
		\end{figure}
		
		\subsection{MIMO Phase-Modulated Channel}
		
		Then we use our derived bounds to evaluate the achievable rate $\const{R}_{\text{total}}$ of the transceiver architecture based on the CP algorithm for the MIMO phase-modulated channels. We consider the scenario with $\nr=4$ and $\nt=1024$, i.e., the number of transmit antennas is much larger than that of receive antennas, thereby facilitating the feasibility of the CP algorithm. We assume a slow-fading channel matrix $\mathsf{H}$ with all entries independently drawn from the complex Gaussian distribution $\mathcal{CN}(0,1)$.
		
		In Fig.~\ref{fig:capacity_rayleigh}, we numerically evaluate the average throughput of the first three sub-channels for the case of $\nr=4$ receive antennas\footnote{As shown in \cite{chen2024when} and \cite{sun2026grouped}, the beamforming scheme can be used in the last sub-channel, which is asymptotically optimal in the low-SNR regime. For this reason, we only consider the first three sub-channels in our analysis.}, for which the differential entropy of the equivalent non-Gaussian noise is numerically computed by the Monte-Carlo method. The total achievable rate is obtained as the sum of the average throughputs of these sub-channels. The gap between our derived bounds is relatively small as the SNR increases, just like in Fig.~\ref{fig:bounds}. A different phenomenon is that, except for the first one, the average throughputs of all subchannel will saturate at high SNRs, which can be attributed to the existence of self-interference brought by the CP algorithm.
		
		\begin{figure}[htbp]
			\centering
			\begin{subfigure}[b]{0.24\textwidth}
				\centering
				\includegraphics[width=\textwidth]{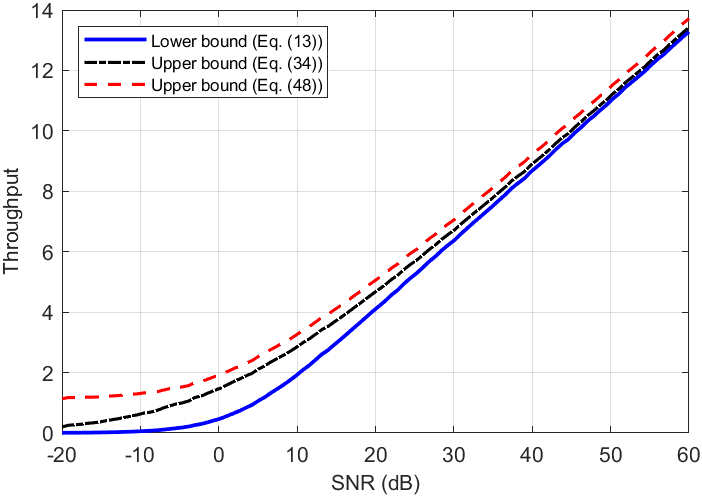}
				\caption{Sub-channel 1.}
				\label{fig:ray_sub1}
			\end{subfigure}
			\hfill
			\begin{subfigure}[b]{0.24\textwidth}
				\centering
				\includegraphics[width=\textwidth]{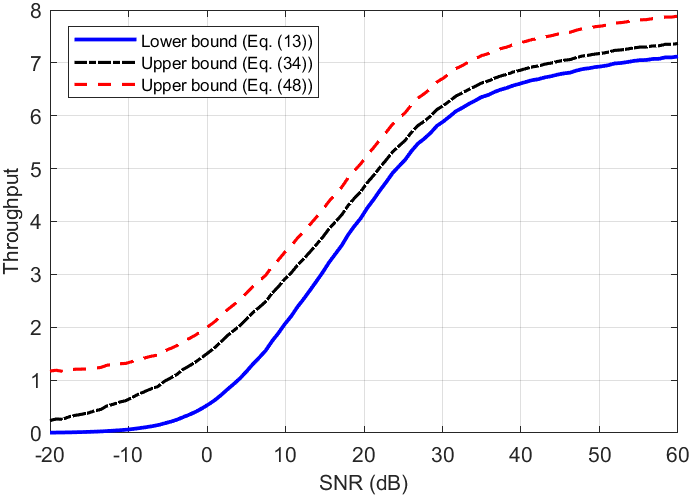}
				\caption{Sub-channel 2.}
				\label{fig:ray_sub2}
			\end{subfigure}
			
			\vspace{0.35cm}
			
			\begin{subfigure}[b]{0.24\textwidth}
				\centering
				\includegraphics[width=\textwidth]{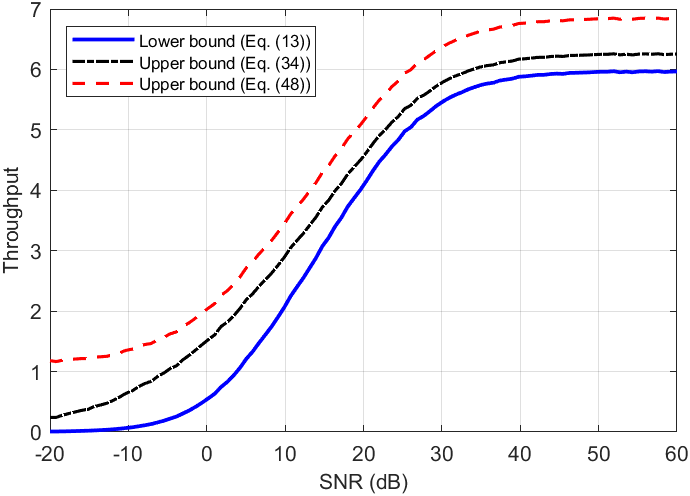}
				\caption{Sub-channel 3.}
				\label{fig:ray_sub3}
			\end{subfigure}
			\hfill
			\begin{subfigure}[b]{0.24\textwidth}
				\centering
				\includegraphics[width=\textwidth]{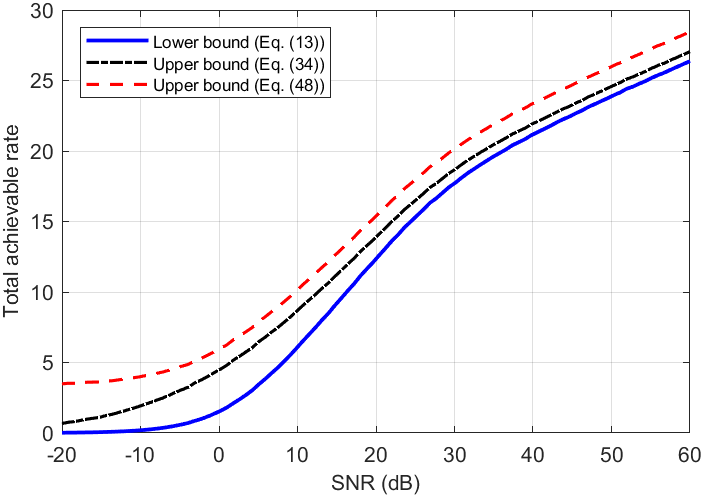}
				\caption{Total achievable rate.}
				\label{fig:ray_total}
			\end{subfigure}
			\caption{Bounds on average throughputs of the transceiver based on the CP algorithm for MIMO phase-modulated channels.}
			\label{fig:capacity_rayleigh}
		\end{figure}
		
		\subsection{RIS-Assisted Symbiotic Communications}
		
		We switch the gears to the achievable rate of a symbiotic communication system with a single transmit antenna, four receive antennas, and an RIS with $1024$ reflecting elements. The channel model follows \cite[Sec. II-B]{sun2026grouped} with the inter-element spacing $d=\lambda/8$. The channel attenuation coefficients are set to $\mu_{\mathrm{LOS}}=-60$ dB (direct path), $\mu_{\mathrm{RR}}=-5$ dB (RIS-to-receiver link), and $\mu_{\mathrm{TR}}=-5$ dB (transmitter-to-RIS link), corresponding to a weak line-of-sight scenario where the RIS-reflected path dominates.
		
		As shown in Fig.~\ref{fig:capacity_ris}, the numerical results for the RIS-assisted channel exhibit a similar trend to the MIMO phase-modulated channels above, but with some differences. Due to the strong spatial correlation caused by the dense RIS element spacing, the CP algorithm efficiently pairs highly correlated columns, resulting in a narrow annulus constellation geometry. In this regime, the upper bound~\eqref{eq:amp_ub} is better than the upper bound~\eqref{eq:disc_ub} at high SNRs. Moreover, it yields a much smaller residual interference variance, so the achievable rate starts to saturate at much higher SNRs. It can be foreseen that with an even larger RIS array, the residual interference can be made negligibly small, making it easier to approach the maximum multiplexing gain.
		
		\begin{figure}[htbp]
			\centering
			\begin{subfigure}[b]{0.24\textwidth}
				\centering
				\includegraphics[width=\textwidth]{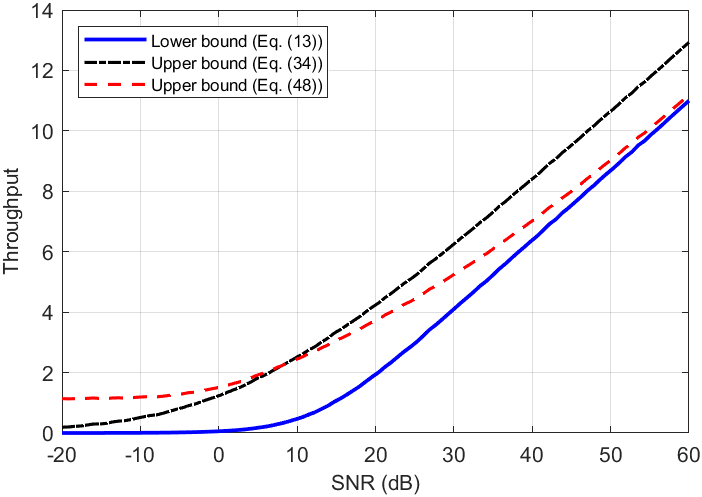}
				\caption{Sub-channel 1.}
				\label{fig:ris_sub1}
			\end{subfigure}
			\hfill
			\begin{subfigure}[b]{0.24\textwidth}
				\centering
				\includegraphics[width=\textwidth]{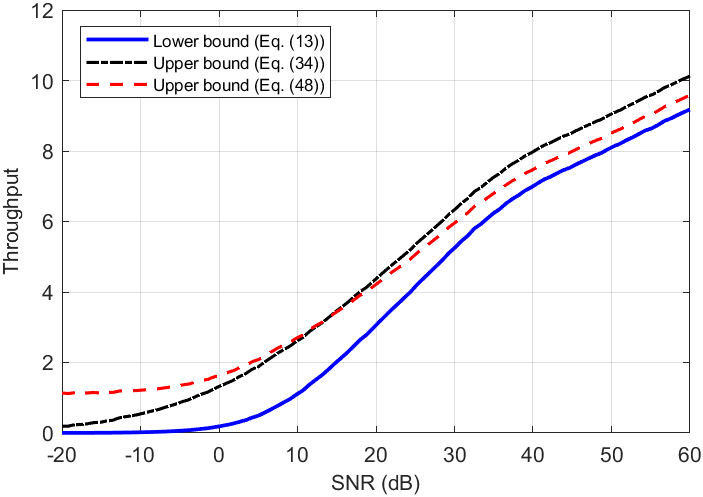}
				\caption{Sub-channel 2.}
				\label{fig:ris_sub2}
			\end{subfigure}
			
			\vspace{0.35cm}
			
			\begin{subfigure}[b]{0.24\textwidth}
				\centering
				\includegraphics[width=\textwidth]{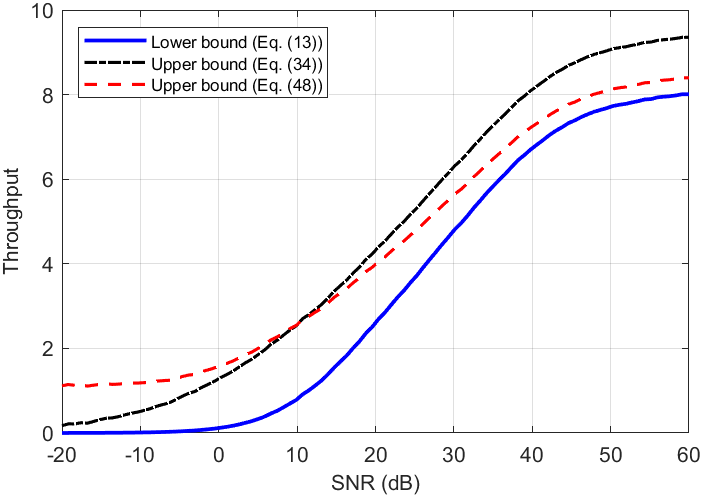}
				\caption{Sub-channel 3.}
				\label{fig:ris_sub3}
			\end{subfigure}
			\hfill
			\begin{subfigure}[b]{0.24\textwidth}
				\centering
				\includegraphics[width=\textwidth]{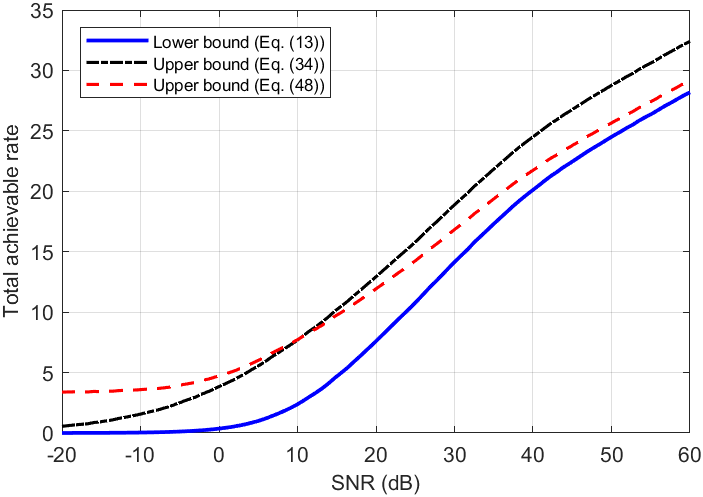}
				\caption{Total achievable rate.}
				\label{fig:ris_total}
			\end{subfigure}
			\caption{Bounds on average throughputs of the transceiver based on the CP algorithm for RIS-assisted symbiotic communication channels.}
			\label{fig:capacity_ris}
		\end{figure}
		
		\section{Conclusion}
		
		In this paper, we evaluated the performance of a specific CP-decomposition-based transceiver architecture for MIMO phase-modulated channels. We derived upper and lower bounds on the achievable rate of each sub-channel under an annular constraint and non-Gaussian interference. The framework was successfully applied to two scenarios: spatially uncorrelated MIMO channels and spatially correlated RIS-assisted channels. The results show that when the number of transmit antennas far exceeds the number of receive antennas, the system achieves higher multiplexing gains over a wider SNR range. Moreover, stronger spatial correlation pushes the interference limit to even higher SNR, enabling almost full multiplexing gain over an extended region.
		
		\bibliographystyle{ieeetr}
		\bibliography{biblio}
		
	\end{document}